\documentclass{lmcs}

\usepackage{amsmath,amssymb,mathtools}
\usepackage{booktabs}
\usepackage{hyperref}

\newcommand{\LAo}{\mathrm{LA}_{\!o}}
\newcommand{\AIK}{\mathsf{AIK}}
\newcommand{\Kon}{K_{\mathrm{on}}^{U}}
\newcommand{\Sing}{I_{\mathrm{sing}}}
\newcommand{\Hom}{\operatorname{Hom}}
\newcommand{\Aut}{\operatorname{Aut}}
\newcommand{\Sh}{\mathbf{Sh}}
\newcommand{\Obs}{\operatorname{Obs}}

\title[Ontological Free Will Beyond P versus NP]{Ontological Free Will as Incompressible Information Adjunction:\
A Noncomputability Boundary Beyond P versus NP}

\author[J. Clech]{J\'er\^ome Clech}[a,b]
\address{Colonel, PhD, Habilitation to Supervise Research (HDR); Chairholder of the Chair of Applied Air and Space Strategies, Centre for Aerospace Strategic Studies (CESA), French Air and Space Force, Paris, France}
\address{Associate Researcher, Technology and Global Affairs Innovation Hub, Paris School of International Affairs (PSIA), Sciences Po, Paris, France}
\email{jerome.clech@sciencespo.fr}

\begin{document}

\begin{abstract}
We study a conditional model of ontological free will in which a pre-act state structures several coherent global continuations without intrinsically distinguishing one as the future actualisation.  The topos-theoretic component is formalised by a choice sheaf and by the action of automorphisms preserving the pre-act data: the absence of a fixed point rules out any natural equivariant selection.  An act changes an unpointed object into a pointed one and reduces its symmetry group to a stabiliser.  Independently, a causal algorithmic incompressibility axiom imposes nearly maximal online description complexity on the actual choices.  Uniformly observable post-act traces then allow their reconstruction.  We prove a conditional transfer result: such traces must carry asymptotically all the singularisation information unavailable in the pre-act regime, and no uniform Turing predictor computes the choices from causal histories alone.  The formalism proves neither that human beings satisfy the model nor a separation of P from NP; it identifies a boundary at which prediction concerns the existence of a computable function rather than its time complexity.
\end{abstract}

\maketitle

\noindent\textbf{Keywords:} topos theory, sheaves, free will, Kolmogorov complexity, algorithmic information, causality, prediction.

\section{Introduction}\label{sec:introduction}

The construction separates two claims that are often conflated.  The first is structural: several continuations satisfy the same local constraints, while no continuation is selected by an intrinsic rule preserving all symmetries of the pre-act data.  The second is algorithmic: the sequence that is actually realised has no uniformly short description relative to the histories available before each act.  The first claim is expressed in the language of sheaves and topoi \cite{MacLaneMoerdijk1992,Johnstone2002}; the second is an independent axiom inspired by algorithmic information theory \cite{Kolmogorov1965,LiVitanyi2019}.

The result is conditional.  It does not derive incompressibility from a qualitative ontology and does not purport to prove that human agents satisfy the proposed axioms.  It establishes the following exact implication.  If a sequence of choices is causally incompressible before the acts and becomes uniformly reconstructible after post-act traces are adjoined, then those traces carry an asymptotically maximal amount of singularisation information.  We maintain throughout a strict separation between assumptions, formal consequences, and ontological interpretation.

The paper proceeds as follows.  \autoref{sec:topos} gives the topos-theoretic obstruction to intrinsic selection.  \autoref{sec:online} defines causal online complexity on an effectively enumerable class of interactive prefix machines.  \autoref{sec:traces} formalises post-act traces.  \autoref{sec:transfer} proves the information-transfer theorem and the necessary complexity of adequate traces.  \autoref{sec:complexity} identifies the limited connection with P versus NP.  \autoref{sec:nonvacuity} provides an explicit non-vacuity model, and \autoref{sec:scope} separates the theorem from its ontological interpretation.

\section{Topos-theoretic structure of possibilities}\label{sec:topos}

\subsection{Choice sheaves and coherent continuations}

For a pre-act history $h$, let $(\mathcal C_h,J_h)$ be a site and let
\[
  \mathcal E_h=\Sh(\mathcal C_h,J_h)
\]
be its topos of sheaves.  An object $F_h\in\mathcal E_h$ represents the determinations compatible with the available data.  For each context $V\in\mathcal C_h$, elements of $F_h(V)$ are local determinations, and restriction maps express compatibility.

The globally coherent continuations are the global elements
\[
  \Gamma(F_h)=\Hom_{\mathcal E_h}(1,F_h).
\]
We assume that $\Gamma(F_h)$ is finite, effectively encodable, and has at least two elements.  Thus the model contains genuine global plurality rather than a lack of constraints.

\subsection{Absence of intrinsic selection}

The informal assertion that no section is ``distinguished'' requires a mathematical criterion.  Let $\mathcal D_h$ denote all relevant pre-act data, and let
\[
  G_h=\Aut(F_h,\mathcal D_h)
\]
be the subgroup of automorphisms of $F_h$ preserving $\mathcal D_h$.  It acts on $\Gamma(F_h)$ by composition.

\begin{defi}[Structural non-singularisation]\label{def:structural}
The pre-act state $h$ satisfies structural non-singularisation when
\[
  \Gamma(F_h)^{G_h}=\varnothing.
\]
Equivalently, no global continuation is fixed by every symmetry preserving the pre-act data.
\end{defi}

\begin{prop}[No intrinsic selector]\label{prop:no-selector}
Under \autoref{def:structural}, there is no morphism of $G_h$-sets
\[
  \sigma_h:\{*\}\longrightarrow\Gamma(F_h),
\]
where $\{*\}$ carries the trivial action.
\end{prop}

\begin{proof}
Equivariance would imply
\[
  g\cdot\sigma_h(*)=\sigma_h(*)
  \qquad\text{for every }g\in G_h.
\]
Hence $\sigma_h(*)$ would belong to $\Gamma(F_h)^{G_h}$, contradicting \autoref{def:structural}.  Notice that $\sigma_h$ is an external map between $G_h$-sets; it is not a morphism internal to $\mathcal E_h$.
\end{proof}

This proposition excludes an intrinsic selector, namely one natural with respect to the symmetries of the model.  It does not by itself exclude an externally encoded algorithm that selects a section arbitrarily.  That stronger exclusion requires the independent incompressibility axiom introduced in \autoref{sec:online}.

\subsection{The act as pointing and symmetry breaking}

Actualising $s_i^*\in\Gamma(F_{h_i})$ turns the unpointed object $F_{h_i}$ into
\[
  (F_{h_i},s_i^*),
  \qquad
  s_i^*:1_{\mathcal E_{h_i}}\longrightarrow F_{h_i},
\]
an object of the coslice category $1\!\downarrow\!\mathcal E_{h_i}$.  The pre-act symmetry group is reduced to the stabiliser
\[
  G_{h_i,s_i^*}=\{g\in G_{h_i}:g(s_i^*)=s_i^*\}.
\]
Topos theory therefore distinguishes two regimes: a structured but unpointed object before the act, and an object equipped with an actualised global point after it.  The framework does not infer algorithmic incompressibility from this structural distinction.

\section{Causal online description complexity}\label{sec:online}

\subsection{Causal interactive machines}

All finite objects are represented by a fixed computable, injective, self-delimiting binary encoding.  A causal interactive prefix machine $M$ receives a self-delimiting program $p$, then successive frames $\langle i,h_i\rangle$, and finally an end marker $\bot$ after stage $n$.  At stage $i$, it must finish exactly one self-delimiting output before it can read $\langle i+1,h_{i+1}\rangle$.  It halts after $\bot$.  A history $h_i$ may contain earlier choices, but contains no trace or datum produced by the current or a later act.

Fix an effective enumeration $(M_e)_{e\in\mathbb N}$ of these machines.  A universal causal machine $U$ receives a prefix code $\langle e,p\rangle$ and simulates $M_e$ while preserving the input-output schedule.

For $S_n=(s_1^*,\ldots,s_n^*)$ and $H_n=(h_1,\ldots,h_n)$, let $\mathcal P_U(S_n\parallel H_n)$ be the set of prefix programs $p$ such that, for every $i\le n$, $U_p$ outputs $s_i^*$ after receiving $\langle i,h_i\rangle$ and before reading the next frame, then halts after $\bot$.

\begin{defi}[Causal online complexity]\label{def:kon}
\[
  \Kon(S_n\parallel H_n)
  =\min\{|p|:p\in\mathcal P_U(S_n\parallel H_n)\}.
\]
The same program must satisfy all stages simultaneously.
\end{defi}

The double bar records an order constraint: this is not ordinary conditional complexity with all of $H_n$ available at the outset.

\begin{lem}[Universality and invariance]\label{lem:invariance}
For every causal interactive prefix machine $M_e$, there is a constant $c_e$, independent of $n,S_n,H_n$, such that
\[
K_{\mathrm{on}}^U(S_n\parallel H_n)
\le K_{\mathrm{on}}^{M_e}(S_n\parallel H_n)+c_e.
\]
Consequently, for two optimal universal machines $U$ and $V$, there is a constant $c_{U,V}$ such that
\[
\left|K_{\mathrm{on}}^U(S_n\parallel H_n)
-K_{\mathrm{on}}^V(S_n\parallel H_n)\right|
\le c_{U,V}.
\]
\end{lem}

\begin{proof}
The machine $U$ simulates $M_e$ from $\langle e,p\rangle$.  The fixed prefix coding $e$ adds $c_e$ bits independently of the transcript, and the simulation preserves the causal lock between consecutive inputs.  Applying the argument in both directions to optimal universal machines yields the second inequality.
\end{proof}

\subsection{Normalisation and the strong model}

Let $C_i(h_i)\subseteq\Gamma(F_{h_i})$ be the finite set of admissible choices at stage $i$, computable from $h_i$, and put $m_i=|C_i(h_i)|$.  Define
\[
  L_n=\log_2(m_1\cdots m_n).
\]
For binary choices, $m_i=2$ and $L_n=n$.

\begin{defi}[Causal algorithmic incompressibility]\label{def:aik}
The process satisfies $\AIK$ if there is a constant $c$ such that, for every $n$,
\[
  \Kon(S_n\parallel H_n)\ge L_n-c.
\]
\end{defi}

\begin{defi}[Strong ontological model]\label{def:laok}
Let $\LAo$ denote global plurality, structural non-singularisation, and pointing by the act.  We define
\[
  \LAo^K:=\LAo+\AIK.
\]
Thus incompressibility is an independent quantitative axiom, not a consequence of $\LAo$ alone.
\end{defi}

\begin{lem}[Compression by a uniform predictor]\label{lem:predictor}
If a uniform program $A$ satisfies $A(h_i)=s_i^*$ for every $i$, then
\[
  \Kon(S_n\parallel H_n)\le |A|+O(1).
\]
\end{lem}

\begin{proof}
At each stage, the universal machine simulates $A$ on the current history and returns its result before accepting the next frame.  The simulator and the protocol have fixed descriptions.
\end{proof}

When $L_n\to\infty$, \autoref{lem:predictor} contradicts $\AIK$.  Uniform algorithmic unpredictability is therefore a consequence of the quantitative axiom rather than a separate clause.

\section{Post-act traces}\label{sec:traces}

\subsection{Causal production protocol}

At each stage we distinguish
\[
 h_i\longrightarrow a_i\longrightarrow(s_i^*,z_i^+)\longrightarrow t_i,
\]
where $a_i$ is the actualising event, $z_i^+$ is the immediate post-act state, and
\[
  t_i=\Obs(z_i^+)
\]
is produced by a fixed computable observation procedure.  The trace is released only after $s_i^*$ has been produced and cannot be an input to its pre-act computation.

\begin{defi}[Admissible trace system]\label{def:traces}
A family $T_n=(t_1,\ldots,t_n)$ is admissible when:
\begin{enumerate}
  \item every $t_i$ is obtained from the post-act state $z_i^+$ by the same procedure $\Obs$;
  \item one computable reconstruction procedure $R$, fixed independently of $i$ and $n$, satisfies
  \[
    R(h_i,t_i)=s_i^*;
  \]
  \item the protocol prevents $R$ and any pre-act predictor from accessing future data.
\end{enumerate}
\end{defi}

The case $t_i=s_i^*$ is allowed: it is an explicit maximal trace.  The issue is not to make the trace mysterious, but to quantify what any adequate trace must carry.

Using the same causal protocol with $(h_i,t_i)$ supplied at stage $i$, define the reconstruction deficiency
\[
  d_n:=\Kon(S_n\parallel H_n,T_n).
\]
A fixed uniform reconstructor gives $d_n=O(1)$, while the parameterised form also covers imperfect traces.

\section{Singularisation information and transfer}\label{sec:transfer}

\begin{defi}[Singularisation information]\label{def:ising}
\[
  \Sing(n):=Kon(S_n\parallel H_n)
  -\Kon(S_n\parallel H_n,T_n).
\]
This quantity is relative to the fixed protocol, encodings, and universal machine, up to an additive constant.
\end{defi}

\begin{thm}[Causal information transfer]\label{thm:transfer}
Assume $\AIK$ and a trace system with reconstruction deficiency $d_n$.  Then
\[
  \Sing(n)\ge L_n-c-d_n.
\]
For uniformly reconstructive traces, $d_n=O(1)$ and hence
\[
  \Sing(n)\ge L_n-O(1).
\]
In the binary case, direct encoding also gives
\[
  \Sing(n)\le n+O(\log n),
\]
so $\Sing(n)=n+O(\log n)$, in the sense of asymptotically maximal transferred information.
\end{thm}

\begin{proof}
The first inequality subtracts the definition of $d_n$ from the $\AIK$ lower bound.  A fixed reconstructor $R$ yields $d_n=O(1)$.  In the binary case, a program may encode the $n$ outputs directly together with a self-delimiting description of their length, using $n+O(\log n)$ bits.  The second term in \autoref{def:ising} is nonnegative, giving the upper bound.
\end{proof}

The theorem is a conditional transfer statement, not an independent derivation of incompressibility.  Its content is the precise comparison of a causally incompressible pre-act regime with a uniformly reconstructible post-act regime.

\begin{prop}[Necessary trace complexity]\label{prop:trace-complexity}
Under $\AIK$ and uniform reconstruction, the traces satisfy
\[
  \Kon(T_n\parallel H_n)\ge L_n-O(1).
\]
\end{prop}

\begin{proof}
Compose any causal program describing $T_n$ from $H_n$ with the fixed reconstructor $R$.  This describes $S_n$ from $H_n$ with constant overhead.  Therefore
\[
  \Kon(S_n\parallel H_n)
  \le \Kon(T_n\parallel H_n)+O(1),
\]
and the claim follows from $\AIK$.
\end{proof}

\begin{cor}[Non-anticipability of traces]\label{cor:traces}
There is no uniform program $G$ satisfying $G(h_i)=t_i$ for every $i$ when $L_n\to\infty$.
\end{cor}

\begin{proof}
Otherwise $h_i\mapsto R(h_i,G(h_i))$ would uniformly predict $s_i^*$, contradicting \autoref{lem:predictor} and $\AIK$.
\end{proof}

\section{Certification and the boundary of complexity theory}\label{sec:complexity}

A relation $R_{\mathsf{GLUE}}(h,s)$ may verify that $s\in\Gamma(F_h)$.  Finding any admissible section is not the same as predicting the section that will be actualised.  The distinction between \textsf{GLUE} and \textsf{SELECT} persists independently of search cost.

Certification of actualisation must be temporally indexed.  If a timeless relation $Q(h,s,\tau)$ already distinguished a unique $s$ for every $h$, it would reintroduce an extensional pre-act selector.  A suitable relation instead depends on the event or its trace, for example $Q(h,t,s,\tau)$.  After the act it may identify a unique result; before the act, $t$ is not part of the available instance.

For the limited connection with P versus NP, suppose that $Q(h,t,s,\tau)$ is decidable in polynomial time and polynomially balanced: for a fixed polynomial $q$,
\[
  |s|+|\tau|\le q(|h|+|t|).
\]
Assume also that every realised post-act instance satisfies
\[
  \exists!s\,\exists\tau\;Q(h,t,s,\tau)=1.
\]
Under P${}={}$NP, the corresponding FP${}={}$FNP consequence and standard bit-by-bit search \cite{Papadimitriou1994} would compute the unique result $M(h,t)=s^*$ in polynomial time.  It would not yield $M(h)=s^*$: the trace belongs to the post-act instance and not to the pre-act input.  P${}={}$NP may reduce the cost of exploiting information contained in a complete instance; it cannot supply a missing input component.  Under $\LAo^K$, a uniform predictor $h_i\mapsto s_i^*$ would contradict $\AIK$ whenever $L_n\to\infty$.

Nothing here separates P from NP, defines a complexity class beyond the computable classes, or resolves the Millennium Problem.  The conclusion is only that, under $\LAo^K$, the pre-act mapping in question is not uniformly computable, so asking whether its running time is polynomial or exponential is inapplicable to that particular object.

\section{Non-vacuity}\label{sec:nonvacuity}

\begin{prop}[Existence of a formal model]\label{prop:existence}
There is a model satisfying structural non-singularisation, $\AIK$, and uniform post-act reconstruction simultaneously.
\end{prop}

\begin{proof}
Take the terminal site, whose sheaf topos is equivalent to $\mathbf{Set}$, and the constant object $F=\{0,1\}$.  Let $G=\{\mathrm{id},\tau\}$, where $\tau$ exchanges $0$ and $1$.  Its action on $\Gamma(F)=\{0,1\}$ has no fixed point, so there is no equivariant map $\{*\}\to\Gamma(F)$.

Fix a computable history $h_0$ independent of the future sequence and set $h_i=h_0$ for every $i$; the stage index is supplied separately in the frame.  Hence $C_i(h_i)=\{0,1\}$ and $L_n=n$.  The successive responses of a causal program $p$ determine at most one binary path.  For each $p$, retain the first prefix it compresses by more than $c$ bits.  The associated cylinder has measure at most $2^{-|p|-c-1}$.  Kraft's inequality gives
\[
 \mu\!\left(\exists n:\Kon(S_n\parallel H_n)<n-c\right)\le 2^{-c}.
\]
For sufficiently large $c$, this set does not cover Cantor space.  Therefore there is an infinite sequence $S=(s_i^*)_{i\ge1}$ such that, for every $n$,
\[
  \Kon(S_n\parallel H_n)\ge n-c.
\]

At stage $i$, the act points the section $s_i^*:1_{\mathbf{Set}}\to F$.  Put $z_i^+=(h_0,s_i^*)$, $t_i=\Obs(z_i^+)=s_i^*$, and $R(h_i,t_i)=t_i$.  The reconstructor is uniform and gives $\Kon(S_n\parallel H_n,T_n)=O(1)$.  All formal conditions hold without the histories encoding the future sequence.
\end{proof}

This construction establishes that $\LAo^K$ is non-vacuous.  Its interpretation as a model of ontological free will is the conceptual choice stated at the outset.

\section{Ontological status and limits}\label{sec:scope}

Three levels must remain separate.
\begin{enumerate}
  \item \emph{Formal level.}  No section is selected by an intrinsic equivariant rule; under $\AIK$, no uniform Turing machine predicts the choices from pre-act histories; reconstructive traces carry asymptotically maximal information.
  \item \emph{Interpretive level.}  In the ontological reading of $\LAo^K$, the absence of structural and algorithmic selection is understood as the pre-act non-existence of singularisation information.  ``Adjunction'' is always relative to the formal causal information state.
  \item \emph{Empirical or metaphysical level.}  The formalism does not prove that information exists nowhere before an act or that human beings satisfy the model.  Applicability requires independent arguments.
\end{enumerate}

The mathematically warranted formulation is therefore: no element of the pre-act state formalised by $h_i$ provides an intrinsic invariant selection of $s_i^*$, and no uniform machine systematically computes $s_i^*$ from $h_i$.  The stronger claim that this information does not yet exist in reality belongs to the substantive interpretation of the model rather than to its formal consequences alone.

\section{Conclusion}\label{sec:conclusion}

The consolidated framework combines five components without conflating them: a plurality of globalisations structured by a sheaf; a topos-theoretic obstruction to intrinsic selection; an independent causal incompressibility axiom; observable and reconstructive post-act traces; and a theorem quantifying the information transferred between the two regimes.

Under $\LAo^K$ and uniform reconstruction, there is neither an intrinsic pre-act selector nor a uniform Turing predictor, and
\[
  \Sing(n)\ge L_n-O(1),
  \qquad
  \Kon(T_n\parallel H_n)\ge L_n-O(1).
\]
For binary choices, singularisation information is asymptotically maximal.  The conclusion is mathematically precise but conditional: it states the consequences of a strong model of ontological free will without turning that model into a metaphysical proof or a solution of P versus NP.

\paragraph*{Acknowledgements.}
Generative AI tools were used as interactive aids for drafting, \LaTeX{} preparation, and consistency checking.  The author conceived and directed the conceptual and mathematical work, verified all arguments and references, and assumes full responsibility for the content.

\bibliographystyle{alphaurl}
\bibliography{references}

\end{document}